\documentclass[cleveref,thm-restate,nolineno]{socg-lipics-v2021}

\pdfoutput=1
\hideLIPIcs

\usepackage[inline]{enumitem}
\usepackage{cases}

\DeclareMathOperator{\plp}{lp}
\DeclareMathOperator{\df}{d_F}

\DeclareMathOperator{\dwf}{d_{wF}}
\DeclareMathOperator{\simpl}{simpl}

\newcommand{\lp}[2]{\ensuremath{\plp\left(#1,#2\right)}}

\def \d{{\mathrm d}}
\def\RR{{\mathbb R}}
\def\NN{{\mathbb N}}
\def\SS{{\mathbb S}}

\let\epsilon\relax
\newcommand{\epsilon}{\varepsilon}

\allowdisplaybreaks

\title{Random projections for curves in high dimensions}

\author{Ioannis Psarros}{Athena Research Center,  Greece}{ipsarros@di.uoa.gr}{https://orcid.org/0000-0002-5079-5003}{The author was partially supported by the EU's Horizon 2020 Research and Innovation programme, under the grant agreement No. 957345: “MORE”.}

\author{Dennis Rohde}{University of Bonn,  Germany}{drohde@uni-bonn.de}{https://orcid.org/0000-0001-8984-1962}{The author was partially supported by the Hausdorff Center for Mathematics.}

\authorrunning{I. Psarros and D. Rohde}

\Copyright{Ioannis Psarros and Dennis Rohde}



\EventEditors{John Q. Open and Joan R. Access}
\EventNoEds{2}
\EventLongTitle{42nd Conference on Very Important Topics (CVIT 2016)}
\EventShortTitle{CVIT 2016}
\EventAcronym{CVIT}
\EventYear{2016}
\EventDate{December 24--27, 2016}
\EventLocation{Little Whinging, United Kingdom}
\EventLogo{}
\SeriesVolume{42}
\ArticleNo{23}

\begin{document}
\setlength{\parindent}{0pt}
\maketitle
\begin{abstract}
    Modern time series analysis requires the ability to handle datasets that are inherently high-dimensional; examples include applications in climatology, where measurements from numerous sensors must be taken into account, or inventory tracking of large shops, where the dimension is defined by the number of tracked items. The standard way to mitigate computational issues arising from the high dimensionality of the data is by applying some dimension reduction technique that preserves the structural properties of the ambient space. The dissimilarity between two time series is often measured by ``discrete'' notions of distance, e.g. the dynamic time warping or the discrete Fr\'echet distance. Since all these distance functions are computed directly on the points of a time series, they are sensitive to different sampling rates or gaps. The continuous Fr\'echet distance offers a popular alternative which aims to alleviate this by taking into account all points on the polygonal curve obtained by linearly interpolating between any two consecutive points in a sequence.

    We study the ability of random projections \`a la Johnson and Lindenstrauss to preserve the continuous Fr\'echet distance of polygonal curves by effectively reducing the dimension. In particular, we show that one can reduce the dimension to $O(\epsilon^{-2} \log N)$, where $N$ is the total number of input points while preserving the continuous Fr\'echet distance between any two determined polygonal curves within a factor of $1\pm \epsilon$. We conclude with applications on clustering.
\end{abstract}
\pagenumbering{gobble}
\clearpage
\pagenumbering{arabic}
\setcounter{page}{1}
\section{Introduction}

Time series analysis lies in the core of various modern applications. Typically, a time series consists of various (physical) measurements over time. Formally, it is a finite sequence of points in $\mathbb{R}^d$. Depending on the use case, the ambient space may be extremely high-dimensional, for example $d \in 2^{\Omega(\log n)}$, or even $d \in 2^{\Omega(n)}$, where $n$ is the number of given sequences. For example, large facilities nowadays supervise their production lines using a plethora of sensors. Another concrete example are climatology applications, where data consist of measurements from multiple sensors, each one corresponding to a different dimension.

Many analysis techniques are based on (dis-)similarity between time series, c.f.~\cite{clustering_time_series}. This is often measured by distance functions such as the Euclidean distance, which however requires the time series to be of same length and does not include any form of alignment between the sequences. This is of course less expressive than distances which are indeed defined over an optimal alignment, e.g.~the dynamic time warping, or the discrete Fr\'echet distance, which are based on Euclidean distances between the points but enable compensation of differences in phase. A common downside of these distances is that they take into account solely the points of a time series. Hence, they are sensitive to differences in sampling rates or data gaps. Here, the continuous Fr\'echet distance offers a popular alternative which aims to alleviate this issue by assuming that time series are discretizations of continuous functions of time. It is an extension of the discrete Fr\'echet distance that takes into account all points on the polygonal curves obtained by linearly interpolating between any two consecutive points in a sequence (where the interpolation is carried out only implicitly).

Two main parameters typically govern computational tasks associated with the Fr\'echet distance: the lengths of the time series and the number of dimensions of the ambient space. In this paper, we study the problem of compressing the input with respect to the latter parameter using a dimension reducing linear transform that preserves Euclidean distances within a factor of $(1\pm \epsilon)$. These transforms, which are usually named Johnson-Lindenstrauss (JL) transforms or embeddings, are a popular tool in dimensionality reduction. The preservation of pairwise distances within a factor of $(1\pm\epsilon)$ is sometimes called JL guarantee. Recent work has provided various probability distributions over JL transforms~\cite{Dasgupta2003,jl_sphericity,Achlioptas03,Linial1995,KaneN14}, which are efficient to sample from and which yield the JL guarantee with at least constant positive probability while the target dimension is only $O(\epsilon^{-2}\log n)$, where $n$ is the size of the input point set. Towards applying this result on time series, one can easily guarantee that all Euclidean distances between points of the time series are preserved. While this has direct implications on ``discrete'' notions of distances between time series, the case of the continuous Fr\'echet distance is far more intriguing.

\subsection{Related Work}

In their seminal paper \cite{johnson_lindenstrauss}, Johnson and Lindenstrauss proved the following statement, which is commonly known as the Johnson-Lindenstrauss lemma and coined the term JL embedding.
\begin{theorem}[\cite{johnson_lindenstrauss}]
\label{JLlemma}
For any $n \in \NN$ and $\epsilon \in (0,1)$ there exists a probability distribution over linear maps $f \colon \RR^d \to \RR^{d^{\prime}}$, where $d^{\prime} \in O(\epsilon^{-2}\log n)$, such that for any $n$-point set $X \subset \mathbb{R}^d$ the following holds with high probability over the choice of $f$: \[\forall p,q \in X:~(1-\epsilon) \| p-q\| \leq \|f(p)-f(q)\| \leq (1+\epsilon) \|p-q\|.\] 
\end{theorem}
In their proof, Johnson and Lindenstrauss \cite{johnson_lindenstrauss} show that this can be achieved by orthogonally projecting the points onto a 
random linear subspace of dimension $O(\epsilon^{-2} \log n)$ -- and indeed there are point sets that require $\Omega(\epsilon^{-2} \log n)$ dimensions \cite{DBLP:journals/dm/Alon03,DBLP:conf/icalp/LarsenN16,DBLP:conf/focs/LarsenN17}. Several proofs of their statement followed, these however don't require a proper projection but only a multiplying the points with a certain random matrix, cf. ~\cite{DBLP:conf/stoc/IndykM98,Dasgupta2003,Achlioptas03,KaneN14,AilonC09}.

The impact of a JL embedding on higher-dimensional objects other than points has already been studied. Magen \cite{DBLP:conf/random/Magen02,DBLP:journals/dcg/Magen07} shows that applying a (scaled) JL embedding not only to a given set $P \subset \mathbb{R}^d$ of points, but to $P \cup W$, where $W \subset \mathbb{R}^d$ is a well-chosen set of points determined by $P$, approximately preserves the height and angles of all triangles determined by any three points in $P$. Magen even extends this result and shows that by a clever choice of $W$, the volume (Lebesgue measure) of the convex hull of any $k-1$ points from $P$ is approximately preserved when the target dimension is in $\Theta(\epsilon^{-2} k \log \lvert P \rvert)$. Furthermore, in this case the distance of any point from $P$ to the affine hull of any $k-1$ other points from $P$ is also approximately preserved. Furthermore, JL embeddings can even be utilized to preserve all pairwise Euclidean and geodesic distances on a smooth manifold \cite{BW09}.

Fr\'echet distance preserving embeddings are a relatively unexplored topic. Recently Driemel and Krivosija \cite{DriemelK18} studied the first Fr\'echet distance preserving embedding for $c$-packed curves, which are curves whose intersections with any ball of radius $r$ are of length at most $cr$. This class of curves was introduced by Driemel et al. \cite{DBLP:journals/dcg/DriemelHW12} and has so far been considered a viable assumption for realistic curves, see e.g. \cite{DBLP:conf/compgeom/AgarwalFPY16,DBLP:conf/focs/Bringmann14,DBLP:journals/siamcomp/DriemelH13}. Driemel and Krivosija consider projections on random lines, where curves are orthogonally projected on a vector which is sampled uniformly at random from the unit sphere. They observed that in any case (even if the curves are not $c$-packed), the discrete Fr\'echet distance between the curves decreases. Furthermore, they show that with high probability the discrete Fr\'echet distance between two curves $\sigma$ and $\tau$, of complexity (number of vertices of the curve) at most $m$, decreases by a factor in $O(m)$. Finally, they proved that there exist $c$-packed curves such that the discrete Fr\'echet distance decreases by a factor in $\Omega(m)$. The latter also holds for the continuous Fr\'echet distance and for the dynamic time warping distance.

More recently, 
Meintrup et al. \cite{DBLP:conf/nips/MeintrupMR19}
studied JL embeddings in the context of preserving the Fr\'echet distance to facilitate $k$-median clustering of curves in a high-dimensional ambient space. They show that when the dimension is reduced to $\Theta(\epsilon^{-2} \log N)$, where $N$ is the total number of vertices of the given curves, the Fr\'echet distances are preserved up to a combined multiplicative error of $(1\pm\epsilon)$ and additive error of $\pm\epsilon L$, where $L$ is the largest arclength of any input curve. For their proof, they only use the JL guarantee, i.e., the $(1\pm\epsilon)$-preservation of Euclidean distances, and properties of the polygonal curves and the Fr\'echet distance, while linearity is not taken into account. In this setting, it seems that the additive error is possible -- Meintrup et al. give a simple example where some vertex-to-vertex distances expand and others contract, which induces an additive error to the Fr\'echet distance. Meintrup et al. complement their results with experimental evaluation showing that in real world data and using a JL transform (which is a \textit{linear} map), the Fr\'echet distance is preserved within the multiplicative error only in almost any case. The other cases can not be distinguished between a failed attempt to obtain a JL embedding (recall the probabilistic nature) and a successful attempt to obtain a JL embedding with the additive error occuring.

\subsection{Our Contributions}

We study the ability of random projections à la Johnson and Lindenstrauss to preserve the continuous Fr\'echet distances among a given set of $n$ polygonal curves, each of complexity (number of vertices of the curve) at most $m$. We show that there exists a set $X$ of vectors (in $\mathbb{R}^d$), of size polynomial in $n$ and $m$ and depending only on the given curves, such that any JL transform for the curves vertices and $X$ also preserves the continuous Fr\'echet distance between any two of the given polygonal curves within a factor of $(1\pm \epsilon)$, without additional additive error. This effectively extends the JL guarantee to pairwise Fr\'echet distances. By plugging in any known JL transform from one of \cite{Dasgupta2003,jl_sphericity,Achlioptas03,Linial1995,KaneN14} we obtain our main dimension reduction result which states that one can reduce the number of dimensions to $O(\epsilon^{-2} \log(nm))$. We achieve our result using a completely different approach than Meintrup et al. \cite{DBLP:conf/nips/MeintrupMR19}. Our approach relies on Fr\'echet distance predicates originating from \cite{DBLP:conf/soda/AfshaniD18}. These allow a reduction from deciding the continuous distance to a finite set of events occurring. Using only the predicates, it is relatively easy to prove that the Fr\'echet distance between two curves does not expand by more than a factor of $(1+\epsilon)$ under a (\textit{linear}) JL transform. To prove that the Fr\'echet distance does not contract by more than a factor of $(1-\epsilon)$ is however much more challenging. We achieve this by proving that \textit{all} distances between one fixed point and any point on a fixed line do not contract by more than a factor of $(1-\epsilon)$ when a JL transform is applied to a well-chosen set of four vectors determined by the point and the line, which is then applied to any vertex of any curve and any line determined by an edge of any curve. We note that this result is comparable to a result by Magen \cite{DBLP:conf/random/Magen02,DBLP:journals/dcg/Magen07}, but our statement is stronger since it takes into account \textit{all} distances between the fixed point and the line and not only the affine distance, i.e., the distance between the point and its orthogonal projection onto the line.

Our motivation is that distance preserving dimensionality reductions imply improved algorithms for various tasks. Best-known algorithms for many proximity problems under the continuous Fr\'echet distance have  exponential dependency on the dimension, in at least one of their performance parameters. Such algorithms either directly employ the continuous Fr\'echet distance, e.g.~the approximation algorithms for $k$-clustering problems~\cite{DBLP:conf/soda/BuchinDR21}, or approximate it with the discrete Fr\'echet distance by resampling the time series to a higher granularity. For example, to the best of our knowledge, the best solution for the approximate near neighbor (ANN) problem in general dimensions derives from building the data structure of Filtser et al. \cite{FFK20}, which originally solves the problem for the discrete Fr\'echet distance, on a modified input. The idea is that a new dense set of vertices can be added to each input polygonal curve so that the discrete Fr\'echet distance of the resulting curves approximates the continuous Fr\'echet distance of the original curves. Under the somewhat restrictive assumption that the arclength of each curve is short, a small number of new vertices suffices. Even in this case though, the space and preprocessing time of the data structure depends exponentially on the number of dimensions. Obviously, polynomial-time algorithms (e.g.~\cite{DBLP:conf/soda/BuchinDGHKLS19}) can also benefit from reducing the number of dimensions, especially when it comes to real applications.    

Our embedding naturally inherits desired properties of the JL transforms like the fact that they are oblivious to the input. This makes it directly applicable to data structure problems like the above-mentioned ANN problem. Moreover, we show that our embedding is also applicable to estimating clustering costs. First, we show that one can approximate the optimal $k$-center cost within a constant factor, with an algorithm that has no dependency on the original dimensionality apart from an initial step of randomly projecting the input curves. Second, we show that one can use any algorithm for computing the $k$-median cost in the dimensionality-reduced space to get a constant factor approximation of the $k$-median cost in the original space. 

\subsection{Organization}

The paper is organized as follows. In \cref{sec:preliminaries} we introduce the necessary notation, definitions and the concept of Fr\'echet distance predicates. In \cref{section:embedding} we prove our main result in two steps. First, as a warm-up, we prove that an application of any JL transformation for the given curves vertices and a polynomial-sized set $X$ determined by these does not increase Fr\'echet distances by more than a factor of $(1+\epsilon)$. In \cref{ss:lowerbound} we prove the challenging part that this also does not decrease Fr\'echet distances by less than a factor of $(1-\epsilon)$. Interestingly, here a different polynomial-sized set $X^\prime$ is used. In \cref{ss:mainresults} we combine both to our main result. Finally, in \cref{section:clustering} we apply our main result to clustering of curves; we modify an existing approximation algorithm for the $(k,\ell)$-center problem (see~\cite{DBLP:conf/soda/BuchinDGHKLS19}) which has negligibly decreased approximation quality compared to the original and we prove that applying \textit{any} algorithm for the $(k,\ell)$-median problem (such as the one from \cite{DBLP:conf/soda/BuchinDR21}) on the embedded curves leads to a constant factor approximation in terms of clustering cost. \cref{section:conclusions} concludes the paper. 

\section{Preliminaries}
\label{sec:preliminaries}

For $n \in \mathbb{N}$ we define $[n] = \{1, \dots, n\}$. By $\lVert \cdot \rVert$ we denote the Euclidean norm, by $\langle \cdot, \cdot \rangle$ we denote the Euclidean dot product and by $\mathbb{S}^{d-1} = \{ p \in \mathbb{R}^d \mid \lVert p \rVert = 1 \}$ we denote the unit sphere in $\mathbb{R}^{d}$. We define line segments, the building blocks of polygonal curves.

\begin{definition}
	A line segment between two points $p_1, p_2 \in \mathbb{R}^d$, denoted by $\overline{p_1p_2}$, is the set of points $\{ (1-\lambda)p_1 + \lambda p_2 \mid \lambda \in [0,1] \}$. For $\lambda \in \mathbb{R}$ we denote by $\lp{\overline{p_1p_2}}{\lambda}$ the point $(1-\lambda)p_1 + \lambda p_2$, lying on the line supporting the segment $\overline{p_1p_2}$.
\end{definition}

We formally define polygonal curves.

\begin{definition}
    \label{def:polygonal_curve}
	A (parameterized) curve is a continuous mapping $\tau \colon [0,1] \rightarrow \mathbb{R}^d$. A curve $\tau$ is polygonal, if and only if, there exist $v_1, \dots, v_m \in \mathbb{R}^d$, no three consecutive on a line, called $\tau$'s vertices and $t_1, \dots, t_m \in [0,1]$ with $t_1 < \dots < t_m$, $t_1 = 0$ and $t_m = 1$, called $\tau$'s instants, such that $\tau$ connects every two consecutive vertices $v_i = \tau(t_i), v_{i+1} = \tau(t_{i+1})$ by a line segment.
\end{definition}

We call the line segments $\overline{v_1v_2}, \dots, \overline{v_{m-1}v_m}$ the edges of $\tau$ and $m$ the complexity of $\tau$, denoted by $\lvert \tau \rvert$. Sometimes we will argue about a sub-curve $\tau[i,j]$ of a given curve $\tau$, which is the polygonal curve determined by the vertices $v_{i}, \dots, v_j$. We define two notions of continuous Fr\'echet distances. We note that the weak Fr\'echet distances is however used only rarely.

\begin{definition}
    Let $\sigma, \tau$ be curves. The weak Fr\'echet distance between $\sigma$ and $\tau$ is \[ \dwf(\sigma,\tau) = \inf_{\substack{f \colon [0,1] \rightarrow [0,1]\\ g \colon [0,1] \rightarrow [0,1]}} \max_{t \in [0,1]} \lVert \sigma(f(t)) - \tau(g(t)) \rVert, \]
    where $f$ and $g$ are continuous functions with $f(0) = g(0) = 0$ and $f(1) = g(1) = 1$.
    The Fr\'echet distance between $\sigma$ and $\tau$ is \[ \df(\sigma, \tau) = \inf_{\substack{f \colon [0,1] \rightarrow [0,1]\\ g \colon [0,1] \rightarrow [0,1]}} \max_{t \in [0,1]} \lVert \sigma(f(t)) - \tau(g(t)) \rVert, \] where $f$ and $g$ are continuous bijections with $f(0) = g(0) = 0$ and $f(1) = g(1) = 1$.
\end{definition}

We define the type of embedding we are interested in. Since we want to keep our results general, we do not specify the target number of dimensions. As a consequence, we drop the JL-terminology and call these $(1\pm\epsilon)$-embeddings.

\begin{definition}
    \label{def:embedding}
	Given a set $P \subset \mathbb{R}^d$ of points and $\epsilon \in (0,1)$, a function $f\colon \mathbb{R}^d \rightarrow \mathbb{R}^{d^\prime}$ is a $(1 \pm \epsilon)$-embedding for $P$, if it holds that \[ \forall p,q \in P: (1-\epsilon) \lVert p - q \rVert \leq \lVert f(p) - f(q) \rVert \leq (1+\epsilon) \lVert p - q \rVert. \]
\end{definition}

We note that if $f$ is linear and $0 \in P$, then  $\forall p \in P$: $(1-\epsilon) \lVert p \rVert \leq \lVert f(p) \rVert \leq (1+\epsilon) \lVert p \rVert$. 

We now define valid sequences with respect to two polygonal curves. Such a sequence can be seen as a discrete skeleton in deciding the continuous Fr\'echet distance and is derived from the free space diagram concept used in Alt and Godau's algorithm \cite{alt_godau}.

\begin{definition}
    Let $\sigma, \tau$ be polygonal curves with vertices $v^\sigma_1, \dots, v^\sigma_{\lvert \sigma \rvert}$, respectively $v^\tau_1, \dots, v^\tau_{\lvert \tau \rvert}$. A valid sequence with respect to $\sigma$ and $\tau$ is a sequence $\mathcal{F} = (i_1, j_1), \dots, (i_k, j_k)$ with
    \begin{itemize}
        \item $i_1 = j_1 = 1$, $i_k = \lvert \sigma \rvert - 1$, $j_k = \lvert \tau \rvert - 1$,
        \item $(i_l, j_l) \in [\lvert\sigma\rvert-1] \times [\lvert\tau\rvert-1]$,
        \item $(i_l-i_{l-1}, j_l-j_{l-1}) \in \{ (0,1), (1,0), (0,-1), (-1,0) \}$ for all $1 < l < k$ and
        \item any pair $(i_l,j_l) \in [\lvert\sigma\rvert-1] \times [\lvert\tau\rvert-1]$ appears at most once in $\mathcal{F}$.
    \end{itemize}
    A valid sequence is said to be monotone if $(i_l-i_{l-1}, j_l-j_{l-1}) \in \{ (0,1), (1,0) \}$ for all $1 < l < k$.
\end{definition}

Further decomposing the free space diagram concept, any valid sequence for two curves $\sigma,\tau$ and any radius $r \geq 0$ induces a set of predicates which truth values in conjunction determine whether $\df(\sigma,\tau) \leq r$, respectively $\dwf(\sigma,\tau) \leq r$.

\begin{definition}[\cite{DBLP:conf/soda/AfshaniD18,DBLP:journals/dcg/DriemelNPP21}]
    \label{def:predicates}
    Let $\sigma, \tau$ be polygonal curves with vertices $v^\sigma_1, \dots, v^\sigma_{\lvert \sigma \rvert}$, respectively $v^\tau_1, \dots, v^\tau_{\lvert \tau \rvert}$ and $r \in \mathbb{R}_{\geq 0}$. We define the Fr\'echet distance predicates for $\sigma$ and $\tau$ with respect to $r$.
    \begin{itemize}
        \item $(P_1)^{\sigma,\tau,r}$: This predicate is true, iff $\lVert \sigma_1 - \tau_1 \rVert \leq r$.
        \item $(P_2)^{\sigma,\tau,r}$: This predicate is true, iff $\lVert \sigma_{\lvert \sigma \rvert} - \tau_{\lvert \tau \rvert} \rVert \leq r$.
        \item $(P_3)^{\sigma,\tau,r}_{(i,j)}$: This predicate is true, iff there exists a point $p \in \overline{v^\sigma_{i}v^\sigma_{i+1}}$ with $\lVert p - v^\tau_{j} \rVert \leq r$.
        \item $(P_4)^{\sigma,\tau,r}_{(i,j)}$: This predicate is true, iff there exists a point $p \in \overline{v^\tau_{j}v^\tau_{j+1}}$ with $\lVert p - v^\sigma_{i} \rVert \leq r$.
        \item $(P_5)^{\sigma,\tau,r}_{(i,j,k)}$: This predicate is true, iff there exist $p_1 = \lp{\overline{v^\sigma_jv^\sigma_{j+1}}}{t_1}$ and $p_2 = \lp{\overline{v^\sigma_jv^\sigma_{j+1}}}{t_2}$ with $\lVert v^\tau_{i} - p_1 \rVert \leq r$, $\lVert v^\tau_k - p_2 \rVert \leq r$ and $t_1 \leq t_2$.
        \item $(P_6)^{\sigma,\tau,r}_{(i,j,k)}$: This predicate is true, iff there exist $p_1 = \lp{\overline{v^\tau_iv^\tau_{i+1}}}{t_1}$ and $p_2 = \lp{\overline{v^\tau_iv^\tau_{i+1}}}{t_2}$ with $\lVert v^\sigma_{j} - p_1 \rVert \leq r$, $\lVert v^\sigma_k - p_2 \rVert \leq r$ and $t_1 \leq t_2$.
    \end{itemize}
\end{definition}

The following two theorems state the aforementioned facts. These will be one of our main tools in obtaining our main results. We note that these are rephrased here to fit our needs.

\begin{theorem}[\cite{DBLP:journals/dcg/DriemelNPP21}]
    \label{theo:predicates_weak_frechet}
    Let $\sigma, \tau$ be polygonal curves and $r \in \mathbb{R}_{\geq 0}$. There exists a valid sequence $\mathcal{F}$ with respect to $\sigma$ and $\tau$, such that $(P_1)^{\sigma,\tau,r} \wedge (P_2)^{\sigma,\tau,r} \wedge \Psi_w^{\sigma,\tau,r}(\mathcal{F})$ is true, where \[ \Psi_w^{\sigma,\tau,r}(\mathcal{F}) = \bigwedge_{\substack{(i,j) \in [\lvert \sigma \rvert] \times [\lvert \tau \rvert]\\(i,j-1),(i,j) \in \mathcal{F}}} (P_3)^{\sigma,\tau,r}_{(i,j)} \bigwedge_{\substack{(i,j) \in [\lvert \tau \rvert] \times [\lvert \sigma \rvert]\\(i-1,j),(i,j) \in \mathcal{F}}} (P_4)^{\sigma,\tau,r}_{(i,j)}, \] if, and only if, $\dwf(\sigma, \tau) \leq r$.
\end{theorem}

\begin{theorem}[\cite{DBLP:conf/soda/AfshaniD18,DBLP:journals/dcg/DriemelNPP21}]
    \label{theo:predicates_frechet}
    Let $\sigma, \tau$ be polygonal curves and $r \in \mathbb{R}_{\geq 0}$. There exists a monotone valid sequence $\mathcal{F}$ with respect to $\sigma$ and $\tau$, such that $(P_1)^{\sigma,\tau,r} \wedge (P_2)^{\sigma,\tau,r} \wedge \Psi^{\sigma,\tau,r}(\mathcal{F})$ is true, where \[ \Psi^{\sigma,\tau,r}(\mathcal{F}) = \hspace{-1.5em} \bigwedge_{\substack{(i,j)\in [\lvert \sigma \rvert] \times [\lvert \tau \rvert]\\(i,j-1),(i,j) \in \mathcal{F}}} \hspace{-1.5em} (P_3)^{\sigma,\tau,r}_{(i,j)} \bigwedge_{\substack{(i,j)\in [\lvert \tau \rvert] \times [\lvert \sigma \rvert]\\(i-1,j),(i,j) \in \mathcal{F}}} \hspace{-1.5em} (P_4)^{\sigma,\tau,r}_{(i,j)} \bigwedge_{\substack{(i,j,k)\in [\lvert \tau \rvert] \times [\lvert \sigma \rvert] \times [\lvert \tau \rvert]\\(i,j-1),(i,k) \in \mathcal{F}\\j < k}} \hspace{-2em} (P_5)^{\sigma,\tau,r}_{(i,j,k)} \bigwedge_{\substack{(i,j,k) \in [\lvert \tau \rvert] \times [\lvert \sigma \rvert] \times [\lvert \sigma \rvert] \\(i-1,j),(k,j) \in \mathcal{F}\\i < k}} \hspace{-2em} (P_6)^{\sigma,\tau,r}_{(i,j,k)},\] if, and only if, $\df(\sigma, \tau) \leq r$.
\end{theorem}

\section{Linear Embeddings Preserve Fr\'echet Distances}
\label{section:embedding}

In this section we prove our main results on embeddings of polygonal curves that approximately preserve the Fr\'echet distance. In the following \cref{lem:upper_bound}, we show that linear $(1\pm \epsilon)$-embeddings for a polynomial number of points determined by the input polygonal curves imply embeddings for the curves that are not expansive by a factor greater than $(1+\epsilon)$. Similarly, in \cref{ss:lowerbound}, we show that linear $(1\pm \epsilon)$-embeddings for a polynomial number of points determined by the curves, imply embeddings for the curves that are not contractive by a factor smaller than $(1-\epsilon)$. Combining these two bounds yields our main results in \cref{ss:mainresults}. Our main dimensionality reduction result states that one can embed a set of $n$ polygonal curves of complexity at most $m$ into a Euclidean space of dimensions $d^\prime \in O(\epsilon^{-2}\log (nm))$, so that all Fr\'echet distances are preserved within a factor of $(1\pm \epsilon)$. The embedding is implemented by mapping the vertices of each polygonal curve with a JL transform. The image of each input curve is a curve in $\RR^{d'}$ having as vertices the images of the original vertices. 

\begin{lemma}
    \label{lem:upper_bound}
    Let $\sigma, \tau$ be polygonal curves with vertices $v^\sigma_1, \dots, v^\sigma_{\lvert \sigma \rvert}$, respectively $v^\tau_1, \dots, v^\tau_{\lvert \tau \rvert}$, and let $f$ be a linear $(1 \pm \epsilon)$-embedding for $P = \{v^\sigma_1, \dots, v^\sigma_{\lvert \sigma \rvert}, v^\tau_1, \dots, v^\tau_{\lvert \tau \rvert} \} \cup P^\prime$, where $P^\prime$ is a set of points determined by $\sigma$ and $\tau$ with $\lvert P^\prime \rvert \in O(\lvert \sigma \rvert^2 \cdot \lvert \tau \rvert + \lvert \tau \rvert^2 \cdot \lvert \sigma \rvert)$.
    Let $\sigma^\prime$ and $\tau^\prime$ be polygonal curves with vertices $f(v^\sigma_1), \dots, f(v^\sigma_{\lvert \sigma \rvert})$, respectively $f(v^\tau_1), \dots, f(v^\tau_{\lvert \tau \rvert})$. It holds that
    \begin{itemize}
        \item $\dwf(\sigma^\prime, \tau^\prime) \leq (1+\epsilon) \dwf(\sigma, \tau)$ and 
        \item $\df(\sigma^\prime, \tau^\prime) \leq (1+\epsilon) \df(\sigma, \tau)$.
    \end{itemize}
\end{lemma}

The proof follows by an application of the $(1 \pm \epsilon)$-embedding to all points determined by the (weak) Fr\'echet distance predicates. The proof can be found in \cref{sec:appendix}.

\subsection{Lower bound}
\label{ss:lowerbound}
In this section, we show that we can use linear $(1\pm \epsilon)$-embeddings for a polynomial number of points determined by the input polygonal curves to define  embeddings for the curves that are not contractive with respect to their Fr\'echet distance by a factor smaller than $(1-\epsilon)$. 

We first introduce a few necessary technical lemmas and then we proceed with the main result. 
We make use of the following lemma, which indicates that inner products are (weakly) concentrated in $(1\pm \epsilon)$-embeddings. Slightly different versions of this lemma have been used before (see e.g.~\cite{DBLP:conf/focs/ArriagaV99,DBLP:journals/jcss/PapadimitriouRTV00,Sarlos06}). Since our statement is a bit more generic, because it holds for any linear $(1\pm \epsilon)$-embedding, and we make use of the involved scaling factors, we include a proof in the appendix for completeness. 

\begin{lemma}
    \label{lem:inner_product}
    Let $f$ be a linear $(1 \pm \epsilon)$-embedding for a finite set $P \subset \mathbb{R}^d$ with $0 \in P$. For all $p,q \in P$ it holds that \[\langle p, q \rangle - 16 \epsilon(\lVert p \rVert \cdot \lVert q \rVert) \leq \langle f(p), f(q) \rangle \leq \langle p, q \rangle + 14 \epsilon(\lVert p \rVert \cdot \lVert q \rVert).\]
\end{lemma}

Next, we prove that $(1\pm \epsilon)$-embeddings for a specific point set do not contract distances between any point on a fixed ray starting from the origin and a fixed point lying in a certain halfspace by a factor smaller than $(1-3\epsilon)$.

\begin{lemma}
\label{lemma:segmentnoncontraction0} 
Let $x \in \RR^d$ and $u\in \SS^{d-1}$ such that $\langle x,u\rangle \leq 0$. 
Let $f$ be a linear $(1 \pm \epsilon/16)$-embedding for $\{0,x,u\}$. For any $\lambda \geq 0$, we have
\[
\left\|f(x)-\lambda \cdot f(u)\right\|  \geq (1-3\epsilon)\|x-\lambda u \|.
\]
\end{lemma}
\begin{proof}
By \cref{lem:inner_product,def:embedding}: 
    \begin{enumerate*}[label=\roman*)]
        \item $\langle f(x),f(u) \rangle \in \langle x, u\rangle \pm \epsilon \|x\|$,
        \item $\|f(x)\| \in (1\pm\epsilon) \|x\|$,
        \item $\|f(u)\| \in (1\pm \epsilon)$.
    \end{enumerate*}
For any  $\lambda \geq 0$ we have:
\begin{align}
    \left\|f(x)-\lambda\cdot f( u) \right\|^2 & = \|f(x)\|^2 + \lambda^2\cdot \|f(u)\|^2  - 2\lambda \cdot \langle f(x),f(u)\rangle \nonumber \\
    &\geq (1-\epsilon)^2 \|x\|^2 +(1-\epsilon)^2 \lambda^2  - 2\lambda \langle x,u\rangle - 2\lambda \epsilon \|x\| \label{eq:sarlosapplication} \\ 
    &\geq  (1-\epsilon)^2 \|x\|^2 +(1-\epsilon)^2\lambda^2  - (1-\epsilon)^2\cdot {2\lambda} \cdot \langle x,u\rangle - 2\lambda \epsilon \cdot \|x\| \label{eq:negativeinner}\\
    &\geq  (1-\epsilon)^2 \|x-\lambda u\|^2 - 2\epsilon\lambda  \|x\|   \nonumber\\
      &\geq  (1-\epsilon)^2 \|x-\lambda u\|^2 - 2\epsilon \|x-\lambda u\|^2 \label{eq:lambdabound}\\
      &\geq (1-3\epsilon)^2 \|x-\lambda u\|^2, \nonumber 
\end{align}
where the last inequality holds, since $\epsilon/16 \in (0,1/4]$. In \cref{eq:sarlosapplication} we use events i), ii), iii), 
in \cref{eq:negativeinner} we use the fact that $\langle x, u \rangle \leq 0$, and in \cref{eq:lambdabound} we use the fact that $\langle x, u \rangle \leq 0$ and $\lambda \geq 0$ implies that $\|x-\lambda u \| \geq \lambda$ and $\|x-\lambda u \| \geq \|x\|$.
\end{proof}

We now prove our main technical lemma. This says that given a fixed line and a fixed point $p$, there is a set $P$ of points such that any linear $(1\pm\epsilon)$-embedding for $P$ does not contract distances between $p$ and any point on the line by a factor smaller than $(1-3\epsilon)$. A somewhat similar statement appears in \cite{DBLP:conf/random/Magen02} which however focuses on the distortion of point-line distances, i.e., how the distance between a point and its orthogonal projection onto the line changes after the embedding.

\begin{lemma}
    \label{lem:affine_line_distance_contraction}
    Let $x,y,z \in \mathbb{R}^d$ and $\ell = \{ \lp{\overline{yz}}{\lambda} \mid \lambda \in \mathbb{R} \}$ be the line supporting $\overline{yz}$. Let $f$ be a linear $(1\pm\epsilon/16)$-embedding for $\{0,u,-u,x-(t+\langle x,u \rangle \cdot u)\}$, where $u \in \mathbb{S}^{d-1}$ and $t \in \mathbb{R}^d$, such that $\langle u, t \rangle = 0$ and $\{ t + \lambda u \mid \lambda \in \mathbb{R} \} = \ell$. For all $\lambda \in \mathbb{R}$ it holds that \[ \lVert f(x) - f(t + \lambda u) \rVert \geq (1-3\epsilon) \lVert x - (t+\lambda u) \rVert.\]
\end{lemma}
\begin{proof}
We first note that such an element $t$ exists, namely the orthogonal projection of $0$ onto $\ell$. Let $p = t + \langle x - t, u \rangle \cdot u = t + \langle x, u \rangle \cdot u$ be the projection of $x$ onto $\ell$ and let $x' = x-p$. Notice that
\[\langle x',u \rangle =\langle x,u \rangle-\langle t+\langle x,u\rangle \cdot u,u\rangle = \langle x,u\rangle-\langle t ,u \rangle - \langle x , u \rangle = 0.\]  We apply \cref{lemma:segmentnoncontraction0} on the vectors $x',u$. This implies that for any $\lambda\geq 0$, 
\begin{align*}
& \|f(x')-\lambda f(u) \| && \geq (1-3\epsilon)\|x' -\lambda u\| \\
\iff & \|f(x-p)-\lambda f(u) \| && \geq (1-3\epsilon)\|x-p -\lambda u\| \\
\iff & \|f(x-(t+\langle x, u \rangle \cdot u))-\lambda f(u) \| && \geq (1-3\epsilon)\|x-(t+\langle x, u \rangle \cdot u) -\lambda u\| \\
\iff & \|f(x)-f(t)-\langle x, u \rangle \cdot f(u)-\lambda f(u) \|&& \geq (1-3\epsilon)\|x-t-\langle x, u \rangle \cdot u -\lambda u\| \
\\
\iff & \|f(x)-f(t+(\langle x, u \rangle +\lambda) \cdot u) \|&& \geq (1-3\epsilon)\|x-(t+(\langle x, u \rangle +\lambda) \cdot u) \|.
\end{align*}
Now by reparametrizing $\lambda^\prime \gets \langle x,u \rangle +\lambda$, we conclude that for any $\lambda^\prime \geq \langle x, u \rangle $,
\begin{align}
\label{eq:largelambdas}
    \|f(x)-f(t+\lambda^\prime \cdot u) \|& \geq (1-3\epsilon)\|x-(t+\lambda^\prime \cdot u) \| .
\end{align}

Finally, we apply \cref{lemma:segmentnoncontraction0} on the vectors $x',-u$. Notice that $\langle x' ,-u\rangle = -\langle x',u \rangle =0$.  This implies that for any $\lambda\geq 0$, 
\begin{align*}
& \|f(x')-\lambda f(-u) \| && \geq (1-3\epsilon)\|x' -\lambda (-u)\|  \\
\iff & \|f(x-p)-\lambda f(-u) \| && \geq (1-3\epsilon)\|x-p -\lambda (-u)\| \\
\iff & \|f(x-(t+\langle x, u \rangle \cdot u))-\lambda f(-u) \| && \geq (1-3\epsilon)\|x-(t+\langle x, u \rangle \cdot u) -\lambda (-u)\| \\
\iff & \|f(x)-f(t)-\langle x, u \rangle \cdot f(u)-\lambda f(-u) \| && \geq (1-3\epsilon)\|x-t-\langle x, u \rangle \cdot u -\lambda(- u)\|
\\
\iff & \|f(x)-f(t+(\langle x, u \rangle -\lambda) \cdot u) \| && \geq (1-3\epsilon)\|x-(t+(\langle x, u \rangle -\lambda) \cdot u) \|.
\end{align*}
Now by reparametrizing $\lambda^\prime \gets \langle x,u \rangle -\lambda$, we conclude that for any $\lambda^\prime \leq \langle x, u \rangle $,
\begin{align}
\label{eq:smalllambdas}
    \|f(x)-f(t+\lambda^\prime \cdot u) \|& \geq (1-3\epsilon)\|x-(t+\lambda^\prime \cdot u) \| .
\end{align}
\cref{eq:largelambdas} and \cref{eq:smalllambdas} conclude the lemma.
\end{proof}

Using the lemma above we can finally prove the main result of this section.

\begin{lemma}
    \label{lem:lower_bound}
    Let $\sigma, \tau$ be polygonal curves with vertices $v^\sigma_1, \dots, v^\sigma_{\lvert \sigma \rvert}$, respectively $v^\tau_1, \dots, v^\tau_{\lvert \tau \rvert}$, and let $f$ be a linear $(1 \pm \epsilon/48)$-embedding for $P = \{v^\sigma_1, \dots, v^\sigma_{\lvert \sigma \rvert}, v^\tau_1, \dots, v^\tau_{\lvert \tau \rvert} \} \cup P^\prime$, where $P^\prime$ is a set of points determined by $\sigma$ and $\tau$ with $\lvert P^\prime \rvert \in O(\lvert \sigma \rvert \cdot \lvert \tau \rvert)$.
        Let $\sigma^\prime$ and $\tau^\prime$ be polygonal curves with vertices $f(v^\sigma_1), \dots, f(v^\sigma_{\lvert \sigma \rvert})$, respectively $f(v^\tau_1), \dots, f(v^\tau_{\lvert \tau \rvert})$. It holds that
    \begin{itemize}
        \item $\dwf(\sigma^\prime, \tau^\prime) \geq (1-\epsilon) \dwf(\sigma, \tau)$ and 
        \item $\df(\sigma^\prime, \tau^\prime) \geq (1-\epsilon) \df(\sigma, \tau)$.
    \end{itemize}
\end{lemma}
\begin{proof}
    For the first claim, let $r = \dwf(\sigma, \tau)$, for the second claim let $r = \df(\sigma, \tau)$. In both cases, let $r^\prime = (1-\epsilon)r$.
    \\
    
    In the following, we prove that for any (monotone) valid sequence $\mathcal{F}$ and any $\delta > 0$ we have that $(P_1)^{\sigma^\prime, \tau^\prime, r^\prime-\delta} \wedge (P_2)^{\sigma^\prime, \tau^\prime, r^\prime-\delta} \wedge \Psi_w^{\sigma^\prime, \tau^\prime, r^\prime-\delta}(\mathcal{F})$, respectively $(P_1)^{\sigma^\prime, \tau^\prime, r^\prime-\delta} \wedge (P_2)^{\sigma^\prime, \tau^\prime, r^\prime-\delta} \wedge \Psi^{\sigma^\prime, \tau^\prime, r^\prime-\delta}(\mathcal{F})$, is false and therefore $\dwf(\sigma^\prime, \tau^\prime) > r^\prime - \delta$, respectively $\df(\sigma^\prime, \tau^\prime) > r^\prime - \delta$ by \cref{theo:predicates_weak_frechet}, respectively \cref{theo:predicates_frechet}.
    
    Now, let $\mathcal{F}$ be an arbitrary (monotone) valid sequence. By definition of $r$ and \cref{theo:predicates_weak_frechet}, respectively \cref{theo:predicates_frechet}, we know that for any $\delta > 0$ it holds that $(P_1)^{\sigma,\tau,r-\delta} \wedge (P_2)^{\sigma,\tau,r-\delta} \wedge \Psi_w^{\sigma,\tau,r-\delta}(\mathcal{F})$, respectively $(P_1)^{\sigma,\tau,r-\delta} \wedge (P_2)^{\sigma,\tau,r-\delta} \wedge \Psi^{\sigma,\tau,r-\delta}(\mathcal{F})$, is false. If $(P_1)^{\sigma,\tau,r-\delta}$ or $(P_2)^{\sigma,\tau,r-\delta}$ is false then clearly $(P_1)^{\sigma^\prime,\tau^\prime,r^\prime-\delta}$ or $(P_2)^{\sigma^\prime,\tau^\prime,r^\prime-\delta}$ is also false by \cref{def:predicates,def:embedding}. In the following, we  assume that $\Psi_w^{\sigma,\tau,r-\delta}(\mathcal{F})$, respectively $\Psi^{\sigma,\tau,r-\delta}(\mathcal{F})$ is false.
    
    Since the arguments for predicates of type $P_3$ and $P_4$ are analogous, we focus on the former type. Assume that $\Psi_w^{\sigma,\tau,r-\delta}(\mathcal{F})$ is false because a predicate $(P_3)^{\sigma,\tau,r-\delta}_{(i,j)}$ is false. This means that there does not exist a point $p \in \overline{v^\sigma_i v^\sigma_{i+1}}$ with $\lVert p - v^\tau_j \rVert \leq r - \delta$. At this point, recall that since $f$ is linear, any points $\lp{\overline{pq}}{t_1}, \dots, \lp{\overline{pq}}{t_n}$, where $p,q \in \mathbb{R}^d$, are still collinear when $f$ is applied and the relative order on the directed lines supporting $\overline{pq}$ is preserved, which is immediate since $f(\lp{\overline{pq}}{t_i}) = \lp{\overline{f(p)f(q)}}{t_i}$. By \cref{lem:affine_line_distance_contraction} for any $t \in \mathbb{R}$ and the determined point $p = \lp{\overline{v^\sigma_i v^\sigma_{i+1}}}{t}$ on the line supporting $\overline{v^\sigma_i v^\sigma_{i+1}}$ it holds that $\lVert f(v^\tau_j) - f(p) \rVert \geq (1-\epsilon) \lVert p - v^\tau_j \rVert$. Thus, for any $f(p) \in \overline{f(v^\sigma_i) f(v^\sigma_{i+1})}$ we have $p \in \overline{v^\sigma_i v^\sigma_{i+1}}$ and $\lVert f(v^\tau_j) - f(p) \rVert \geq (1-\epsilon) \lVert p - v^\tau_j \rVert$, which in conclusion is larger than $r^\prime - \delta$, hence $(P_3)^{\sigma^\prime,\tau^\prime,r^\prime-\delta}_{(i,j)}$ is false and therefore $\Psi_w^{\sigma^\prime,\tau^\prime,r^\prime-\delta}(\mathcal{F})$ is false. The first claim follows by \cref{theo:predicates_weak_frechet}.
    
    Now, since again the arguments for predicates of type $P_5$ and $P_6$ are also analogous, we focus on the former. Assume that $\Psi^{\sigma, \tau, r - \delta}(\mathcal{F})$ is false, because a predicate $(P_5)^{\sigma,\tau,r-\delta}_{(i,j,k)}$ is false. This means that for any two $t_1, t_2 \in \mathbb{R}$ with $t_1 \leq t_2$, the points $p_1 = \lp{\overline{v^\sigma_j v^\sigma_{j+1}}}{t_1}$ and $p_2 = \lp{\overline{v^\sigma_j v^\sigma_{j+1}}}{t_2}$ do not satisfy $\lVert v^\tau_i - p_1 \rVert \leq r - \delta$ or $\lVert v^\tau_k - p_2 \rVert \leq r - \delta$. Since by \cref{lem:affine_line_distance_contraction} we have $\lVert f(v^\tau_i) - f(p_1) \rVert \geq (1-\epsilon) \lVert v^\tau_i - p_1 \rVert$ and $\lVert f(v^\tau_k) - f(p_2) \rVert \geq (1-\epsilon) \lVert v^\tau_k - p_2 \rVert$, one of these distances must be larger than $r^\prime - \delta$ and it follows that $(P_5)^{\sigma^\prime, \tau^\prime, r^\prime - \delta}_{(i,j,k)}$ is false. Therefore, $\Psi^{\sigma^\prime, \tau^\prime, r^\prime-\delta}(\mathcal{F})$ is false and the second claim follows by \cref{theo:predicates_frechet}.
    
    Finally, for the above statements to hold, the set $P^\prime$ contains $0$, both directions $u,-u \in \mathbb{S}^{d-1}$ determined by an edge of $\sigma$ or $\tau$ and all points $x - (t+\langle x, u \rangle \cdot u)$, where $x$ is a vertex of a curve $\sigma$ or $\tau$, and $t,u$ determine a line supporting an edge of $\tau$ or $\sigma$.
\end{proof}

\subsection{Main result}

\label{ss:mainresults}
We now prove our main result which combines the upper and lower bounds on the distortion and \cref{lem:upper_bound}. 

\begin{theorem}
    \label{coro:JLembedding}
    Let $T = \{ \tau_1, \dots, \tau_n \}$ be a set of polygonal curves in $\mathbb{R}^d$, each of complexity at most $m$. There exists a probability distribution over linear maps $f \colon \RR^d \to \RR^{d^{\prime}}$, where $d^{\prime} \in O(\epsilon^{-2}\log (nm))$, 
    such that with high probability over the choice of $f$, the following is true for all $\sigma, \tau \in T$:
    \begin{itemize}
        \item $\lvert \dwf(\sigma, \tau) - \dwf(F(\sigma), F(\tau)) \rvert \leq \epsilon \cdot \dwf(\sigma, \tau)$ and 
        \item $\lvert \df(\sigma, \tau) - \df(F(\sigma), F(\tau)) \rvert \leq \epsilon \cdot \df(\sigma, \tau)$,
    \end{itemize}
    where for any $\tau \in T$ with vertices $v^\tau_1, \dots, v^\tau_{\lvert \tau \rvert}$ we let $F(\tau)$ be the curve with vertices $f(v^\tau_1), \dots, f(v^\tau_{\lvert \tau \rvert})$.
\end{theorem}
\begin{proof}
We apply \cref{JLlemma} on the set $P$ of $O(n^2 m^3)$ points determined by an application of \cref{lem:upper_bound,lem:lower_bound} on all pairs of curves in $T$.
\end{proof}

\section{Application to Clustering}
\label{section:clustering}

In this section, we study the 
effect of randomized $(1\pm\epsilon)$-embeddings on the cost of $k$-clustering of polygonal curves. In particular, we show that a constant factor approximation of the cost of the optimal $k$-center solution can be computed  with an algorithm, which, except for the time needed to embed the input curves, runs in time independent of the input dimensionality. Moreover, we show that the optimal cost of the $k$-median problem is preserved within a constant factor in the target space. This means that running any algorithm for the $k$-median problem in the target space, yields an algorithm for estimating the cost in the original space.

This effectively reduces the computational effort required for approximating the clustering cost, and it directly assists analytical tasks like estimating the optimal number of clusters -- where cost estimations for multiple values of $k$ are typically performed.

\subsection{Clustering Under the Fr\'echet Distance}

In 2016, Driemel et al. \cite{DBLP:conf/soda/DriemelKS16} introduced clustering under the Fr\'echet distance, for the purpose of clustering (one-dimensional) time series. The objectives, named $(k,\ell)$-center and $(k,\ell)$-median, are derived from the well-known $k$-center and $k$-median objectives in Euclidean $k$-clustering. Both are $\mathrm{NP}$-hard \cite{DBLP:conf/soda/DriemelKS16,DBLP:conf/soda/BuchinDGHKLS19,DBLP:conf/swat/BuchinDS20}, even if $k=1$ and $d=1$, and the $(k,\ell)$-center problem is even $\mathrm{NP}$-hard to approximate within a factor of $(2.25 - \epsilon)$ in general dimensions~\cite{DBLP:conf/soda/BuchinDGHKLS19}. One particularity of these clustering approaches is that the obtained center curves should be of low complexity. In detail, while the given curves have complexity at most $m$ each, the centers should be of complexity at most $\ell$ each, where $\ell \ll m$ is a constant. The idea behind is that due to the linear interpolation, a \textit{compact} summary of the cluster members through an aggregate center curve is enabled. A nice side effect is that overfitting, which may occur without the complexity restriction, is suppressed. For further details see~\cite{DBLP:conf/soda/DriemelKS16}. 

We now present a modification of the constant factor approximation algorithm for $(k,\ell)$-center clustering from \cite{DBLP:conf/soda/BuchinDGHKLS19}. We note that due to its appealing complexity, this algorithm is used vastly in practice (c.f. \cite{DBLP:conf/gis/BuchinDLN19}) and therefore constitutes a prime candidate to be combined with dimensionality reduction.

\subsection{\texorpdfstring{$(k,\ell)$}{(k,l)}-Center Clustering}
\label{ss:center}

We formally define the $(k,\ell)$-center clustering objective.

\begin{definition}
    The $(k,\ell)$-center clustering problem is to compute a set $C$ of $k$ polygonal curves in $\mathbb{R}^d$, of complexity at most $\ell$ each, which minimizes the cost $\max_{\tau \in T} \min_{c\in C} \df(\tau, c)$, where $T = \{ \tau_1, \dots, \tau_n \}$ is a given set of polygonal curves in $\mathbb{R}^d$ of complexity at most $m$ each, and $k \in \mathbb{N}, \ell \in \mathbb{N}_{\geq 2}$ are constant parameters of the problem.
\end{definition}

The following algorithm largely makes use of simplifications of input curves. We formally define this concept.

\begin{definition}
    An $\alpha$-approximate minimum-error $\ell$-simplification of a curve $\tau$ in $\mathbb{R}^d$ is a curve $\sigma = \simpl(\tau)$ in $\mathbb{R}^d$ with at most $\ell$ vertices, where $\ell \in \mathbb{N}_{\geq 2}$ and $\alpha \geq 1$ are given parameters, such that $\df(\tau, \sigma) \leq \alpha \cdot \df(\tau, \sigma^\prime)$ for all other curves $\sigma^\prime$ with $\ell$ vertices.
\end{definition}

A simplification $\sigma = \simpl(\tau)$ is vertex-restricted if the sequence of its vertices is a subsequence of the sequence of $\tau$s vertices. Crucial in our modification of the algorithm by Buchin et al. \cite{DBLP:conf/soda/BuchinDGHKLS19} is that we want to compute simplifications in the dimensionality-reduced ambient space to spare running time. In the following, we give a thorough analysis of the effect of dimensionality reduction before simplification. The proof can be found in \cref{sec:appendix}.

\begin{theorem}
    \label{theo:simpl_embedding}
    Let $F$ be the embedding of \cref{coro:JLembedding} with parameter $\epsilon \in (0,1/2]$, for a given set $T$ of $n$ polygonal curves in $\mathbb{R}^d$ of complexity at most $m$ each, all segments $\overline{v^\tau_i v^\tau_j}$, all subcurves $\tau[i,j]$ as well as all vertex-restricted $\ell$-simplifications of all $\tau \in T$ (where $v^\tau_1, \dots, v^\tau_{\lvert \tau \rvert}$ are the vertices of $\tau$ and $i,j \in [\lvert \tau \rvert]$ with $i < j$). For each $\tau \in T$, a $4$-approximate minimum-error $\ell$-simplification $\simpl(F(\tau))$ of $F(\tau)$ can be computed in time $O(d^\prime \cdot \lvert \tau \rvert^3 \log \lvert \tau \rvert)$ and for all $\sigma \in T$ it holds that \[(1-\epsilon) \df(\sigma, \simpl(\tau)) \leq \df(F(\sigma), \simpl(F(\tau)) \leq (1+\epsilon) \df(\sigma, \simpl(\tau)),\]where $\simpl(\tau)$ denotes a $(4+16\epsilon)$-approximate minimum-error $\ell$-simplification of $\tau$.
\end{theorem}

We now present our modification of the algorithm. Let $F$ denote the embedding from \cref{coro:JLembedding} for $T \cup T^\prime \cup C^\ast$, where $T^\prime$ is the set of all segments $\overline{v^\tau_i v^\tau_j}$, all subcurves $\tau[i,j]$ as well as all vertex-restricted $\ell$-simplifications of all $\tau \in T$ (where $v^\tau_1, \dots, v^\tau_{\lvert \tau \rvert}$ are the vertices of $\tau$ and $i,j \in [\lvert \tau \rvert]$ with $i < j$), and $C^\ast$ is an optimal set of $k$ centers for $T$.

The algorithm first sets $C = \{ \simpl(F(\tau)) \}$ for an arbitrary $\tau \in T$. Then, until $\lvert C \rvert = k$ it computes a curve $\tau \in T$ that maximizes $\min_{c \in C} \df(F(\tau), c)$ and sets $C = C \cup \{ \simpl(F(\tau)) \}$. Finally, it returns $C$.

We now prove the approximation guarantee and analyse the running time of this algorithm, thereby we adapt parts of the analysis in \cite{DBLP:conf/soda/BuchinDGHKLS19}. The proof can be found in \cref{sec:appendix}.

\begin{theorem}
    \label{theo:klcenter}
    Given a set $T$ of $n$ polygonal curves in $\mathbb{R}^d$ of complexity at most $m$ each, and a parameter $\epsilon \in (0,1/2]$, the above algorithm returns a solution $C$ to the $(k,\ell)$-center clustering problem, consisting of $k$ curves in $\mathbb{R}^{O(\epsilon^{-2}\ell\log(knm))}$ of complexity at most $\ell$ each, such that \[ (1-3\epsilon) r^\ast \leq \max_{\tau \in T} \min_{c \in C} \df(F(\tau), c) \leq (6+38\epsilon) r^\ast,\]
    where $r^\ast$ denotes the cost of an optimal solution. The algorithm has running time \[O(\epsilon^{-2} k \ell \log(nm+k) m^3 \log m + \epsilon^{-2} \ell \log(nm+k) k^2 nm \log m).\]
\end{theorem}

\subsection{\texorpdfstring{$(k,\ell)$}{(k,l)}-Median Clustering}

In this section, we show that the cost of the optimal $(k,\ell)$-median solution is preserved within a constant factor, when projecting the input curves as described in \cref{section:embedding}. We first define the $(k,\ell)$-median clustering problem. 

\begin{definition}
    The $(k,\ell)$-median clustering problem is to compute a set $C$ of $k$ polygonal curves in $\mathbb{R}^d$ of complexity at most $\ell$ each, which minimizes the cost $\sum_{\tau \in T} \min_{c\in C} \df(\tau, c)$, where $T = \{ \tau_1, \dots, \tau_n \}$ is a given set of polygonal curves in $\mathbb{R}^d$ of complexity at most $m$ each, and $k \in \mathbb{N},\ell \in \mathbb{N}_{\geq 2}$ are constant parameters of the problem. 
\end{definition}

In \cref{sssunrestrictedmedians}, we focus on the case $\ell \geq m$, and we bound the distortion of the optimal cost by a factor of $2+O(\epsilon)$. 
In \cref{sssrestrictedmedians}, we discuss case $\ell < m$, and we bound the distortion of the optimal cost by a factor of $6+O(\epsilon)$. 

\subsubsection{Unrestricted medians}
\label{sssunrestrictedmedians}
In this section, we present our results on the $(k,\ell)$-median clustering problem, when $\ell\geq m$. 
 Computing medians of complexity $\ell=m$ is a widely accepted scenario following, for example,  from the wide acceptance of local search methods for clustering, which explore candidate solutions from the set of input curves. The proof follows a similar reasoning as in \cref{sssrestrictedmedians} and is diverted to \cref{appendixrestrictedmedians}. Comparing to \cref{sssrestrictedmedians}, we obtain an improved bound on the approximation factor. This is mainly because simplifications are no longer needed in order to obtain a meaningful bound. 
\begin{restatable}{theorem}{unrestrictedcorollary}
Let $T = \{ \tau_1, \dots, \tau_n \}$ be a set of polygonal curves in $\mathbb{R}^d$ of complexity at most $m$ each and let $\ell \geq m $. There exists a probability distribution over linear maps $f \colon \RR^d \to \RR^{d^{\prime}}$, where $d^{\prime} \in O(\epsilon^{-2}\log (n\ell))$, such that with high probability over the choice of $f$, the following is true. For any polygonal curve $\tau$ with vertices $v^\tau_1, \dots, v^\tau_{\lvert \tau \rvert}$, we define $F(\tau)$ to be the curve with vertices $f(v^\tau_1), \dots, f(v^\tau_{\lvert \tau \rvert})$.
Then, 
\[
\frac{1-\epsilon}{2}\cdot r^{\ast}\leq  r_f^{\ast} \leq (1+\epsilon)\cdot r^{\ast},
\]
where $r^{\ast}$ is the cost of an  optimal solution to the $(k,\ell)$-median problem on $T$, and $r_f^{\ast}$ is the cost of an optimal solution to the $(k,\ell)$-median problem on $F(T)$.  
\end{restatable}

\subsubsection{Restricted medians}
\label{sssrestrictedmedians}

To bound the cost of the optimal $(k,\ell)$-median in the projected space, we use the notion of simplifications which was introduced in Section~\ref{ss:center}. By an averaging argument, for each cluster, there exists  an input curve $\sigma_i$ which is within distance $ \frac{1}{|T_i|}\cdot \sum_{\tau\in T_i}\df(F(\tau),c_i^f)$ from the optimal median $c_i^f$, where $T_i$ is the input curves associated with the $i$\textsuperscript{th} cluster in the projected space. To lower bound the optimal cost in the projected space, we repeatedly apply the triangle inequality on distances involving a vertex-restricted $\ell$-simplification of $\sigma_i$ and a vertex-restricted 
 $\ell$-simplification of $F(\sigma_i)$. The upper bound simply follows by the non-contraction guarantee of  JL transforms, on distances between input curves and the optimal medians in the original space. 
The complete proof can be found in  \cref{appendixrestrictedmedians}.
\begin{restatable}{theorem}{restrictedcorollary}
Let $T = \{ \tau_1, \dots, \tau_n \}$ be a set of polygonal curves in $\mathbb{R}^d$ of complexity at most $m$ each and let $\ell < m $. There exists a probability distribution over linear maps $f \colon \RR^d \to \RR^{d^{\prime}}$, where $d^{\prime} \in O(\epsilon^{-2} \ell  \log (nm))$, such that with high probability over the choice of $f$, the following is true. For any polygonal curve $\tau$ with vertices $v^\tau_1, \dots, v^\tau_{\lvert \tau \rvert}$, we define $F(\tau)$ to be the curve with vertices $f(v^\tau_1), \dots, f(v^\tau_{\lvert \tau \rvert})$.
Then,  
\[
\frac{1-\epsilon}{6\cdot(1+\epsilon)}\cdot r^{\ast}\leq  r_f^{\ast} \leq (1+\epsilon)\cdot r^{\ast},
\]
where $r^{\ast}$ is the cost of an  optimal solution to the $(k,\ell)$-median problem on $T$, and $r_f^{\ast}$ is the cost of an optimal solution to the  $(k,\ell)$-median problem on $F(T)$.  
\end{restatable}

\section{Conclusion}
\label{section:conclusions}

Our results are in line with the results by Magen \cite{DBLP:conf/random/Magen02,DBLP:journals/dcg/Magen07} in the sense that by increasing the constant hidden in the $O$-notation specifying the number of dimensions of the dimensionality-reduced space, JL transforms become more powerful and do not only preserve pairwise Euclidean distances but also affine distances, angles and volumes, and as we have proven, Fr\'echet distances.

Concerning JL transforms we have improved the work by Meintrup et al. \cite{DBLP:conf/nips/MeintrupMR19} by proving that no additive error is involved in the resulting Fr\'echet distances. To facilitate this result, we had to incorporate the linearity of these transforms, which is not done in \cite{DBLP:conf/nips/MeintrupMR19}. Interestingly, this shows that when one uses a terminal embedding instead (see e.g. \cite{DBLP:conf/stoc/NarayananN19}) -- for example to handle a dynamic setting involving queries -- this may induce an additive error to the Fr\'echet distance, as the results by Meintrup et al. \cite{DBLP:conf/nips/MeintrupMR19} can still be applied but ours can not since terminal embeddings are non-linear. Consequently, in contrast to Euclidean distances where a terminal embedding constitutes a proper extension of a JL embedding, this may not be the case when it comes to Fr\'echet distances.

One open question of practical importance is whether one can improve our result for polygonal curves that satisfy some realistic structural assumption, e.g., $c$-packness~\cite{DBLP:journals/dcg/DriemelHW12}. Moreover, it is possible that our implications on clustering can be improved. One question there is whether one can reduce (or eliminate) the dependence on $n$ from the target dimension, in the same spirit as with the analogous results for the Euclidean distance~\cite{MMR19}.

\bibliographystyle{plainurl}
\bibliography{literature}

\appendix

\section{Missing Proofs}
\label{sec:appendix}

\begin{proof}[Proof of \cref{lem:upper_bound}]
    For the first claim, let $r = \dwf(\sigma, \tau)$, for the second claim let $r = \df(\sigma, \tau)$. In both cases, let $r^\prime = (1+\epsilon)r$. Since $\sigma$ and $\tau$ have (weak) Fr\'echet distance $r$, there exists a (monotone) valid sequence $\mathcal{F}$, such that $(P_1)^{\sigma,\tau,r} \wedge (P_2)^{\sigma,\tau,r} \wedge \Psi_w^{\sigma,\tau,r}(\mathcal{F})$, respectively $(P_1)^{\sigma,\tau,r} \wedge (P_2)^{\sigma,\tau,r} \wedge \Psi^{\sigma,\tau,r}(\mathcal{F})$, is true, by \cref{theo:predicates_weak_frechet}, respectively \cref{theo:predicates_frechet}.
    \\
    
    Clearly, since $f$ is a $(1\pm \epsilon)$-embedding, $(P_1)^{\sigma^\prime,\tau^\prime,r^\prime}$ and $(P_2)^{\sigma^\prime,\tau^\prime,r^\prime}$ are both true.
    \\
    
    We now denote \[\mathcal{I}_3 = \{ (i,j) \mid (i,j-1),(i,j) \in \mathcal{F} \},\ \mathcal{I}_4 = \{ (i,j) \mid (i-1,j),(i,j) \in \mathcal{F} \} \] and \[ \mathcal{I}_5 = \{ (i,j,k) \mid (i,j-1),(i,k) \in \mathcal{F}, j < k \},\ \mathcal{I}_6 = \{ (i,j,k) \mid (i-1,j),(k,j) \in \mathcal{F}, i < k \}.\] Furthermore, let $\mathcal{P}_3$ be the set of points $p \in \overline{v^\sigma_i v^\sigma_{i+1}}$ with $\lVert p - v^\tau_j \rVert \leq r$ guaranteed by $(P_3)^{\sigma,\tau,r}_{(i,j)}$, for all $(i,j) \in \mathcal{I}_3$, $\mathcal{P}_4$ be the set of points $p \in \overline{v^\tau_jv^\tau_{j+1}}$ with $\lVert p - v^\sigma_i \rVert \leq r$ guaranteed by $(P_4)^{\sigma,\tau,r}_{(i,j)}$, for all $(i,j) \in \mathcal{I}_4$, $\mathcal{P}_5$ be the set of points $p_1 = \lp{\overline{v^\sigma_jv^\sigma_{j+1}}}{t_1}, p_2 = \lp{\overline{v^\sigma_jv^\sigma_{j+1}}}{t_2}$ with $\lVert v^\tau_i - p_1 \rVert \leq r$, $\lVert v^\tau_k - p_2 \rVert \leq r$ and $t_1 \leq t_2$ guaranteed by $(P_5)^{\sigma,\tau,r}_{(i,j,k)}$, for all $(i,j,k) \in \mathcal{I}_5$, and $(P_6)^{\sigma,\tau,r}_{(i,j,k)}$ be the set of points $p_1 = \lp{\overline{v^\tau_i v^\tau_{i+1}}}{t_1}, p_2 = \lp{\overline{v^\tau_i v^\tau_{i+1}}}{t_2}$ with $\lVert v^\sigma_j - p_1 \rVert \leq r$, $\lVert v^\sigma_k - p_2 \rVert \leq r$ and $t_1 \leq t_2$ guaranteed by $(P_6)^{\sigma,\tau,r}_{(i,j,k)}$ for all $(i,j,k) \in \mathcal{I}_6$.
    \\
    
    We let $P^\prime = \mathcal{P}_3 \cup \mathcal{P}_4 \cup \mathcal{P}_5 \cup \mathcal{P}_6$. Clearly, for any $(i,j) \in \mathcal{I}_3$, respectively $(i,j) \in \mathcal{I}_4$, $(P_3)^{\sigma^\prime,\tau^\prime, r^\prime}_{(i,j)}$, respectively $(P_4)^{\sigma^\prime,\tau^\prime, r^\prime}_{(i,j)}$ is true, since there exist points $p,q \in f(P^\prime)$ with $p \in \overline{f(v^\sigma_i)f(v^\sigma_{i+1})}$ and $\lVert p - f(v^\tau_j) \rVert \leq r^\prime$, respectively $q \in \overline{f(v^\tau_j)f(v^\tau_{j+1})}$ and $\lVert p - f(v^\sigma_i) \rVert \leq r^\prime$. At this point, it can be observed that $(P_1)^{\sigma^\prime,\tau^\prime,r^\prime} \wedge (P_2)^{\sigma^\prime,\tau^\prime,r^\prime} \wedge \Psi_w^{\sigma^\prime,\tau^\prime, r^\prime}(\mathcal{F})$ is true and therefore the first claim follows by \cref{theo:predicates_weak_frechet}.
    \\
    
    For the following, observe that since $f$ is linear, any points $\lp{\overline{pq}}{t_1}, \dots, \lp{\overline{pq}}{t_n}$, where $p,q \in \mathbb{R}^d$, are still collinear when $f$ is applied and the relative order on the directed lines supporting $\overline{pq}$ is preserved -- this is immediate since $f(\lp{\overline{pq}}{t_i}) = \lp{\overline{f(p)f(q)}}{t_i}$. We conclude that for any $(i,j,k) \in \mathcal{P}_5$, respectively $(i,j,k) \in \mathcal{P}_6$, $(P_5)^{\sigma^\prime,\tau^\prime,r^\prime}_{(i,j,k)}$, respectively $(P_6)^{\sigma^\prime,\tau^\prime,r^\prime}_{(i,j,k)}$, is true, since there exists points $p_1,p_2,q_1,q_2 \in f(P^\prime)$ with $p_1 = \lp{\overline{f(v^\sigma_j)f(v^\sigma_{j+1})}}{t_1}$, $p_2 = \lp{\overline{f(v^\sigma_j)f(v^\sigma_{j+1})}}{t_2}$, $\lVert f(v^\tau_i) - p_1 \rVert \leq r^\prime$, $\lVert f(v^\tau_k) - p_2 \rVert \leq r^\prime$ and $t_1 \leq t_2$, respectively $q_1 = \lp{\overline{f(v^\tau_i)f(v^\tau_{i+1})}}{t_3}$, $q_2 = \lp{\overline{f(v^\tau_i)f(v^\tau_{i+1})}}{t_4}$, $\lVert f(v^\sigma_j) - q_1 \rVert \leq r^\prime$, $\lVert f(v^\sigma_k) - q_2 \rVert \leq r^\prime$ and $t_3 \leq t_4$. Thus, $(P_1)^{\sigma^\prime,\tau^\prime,r^\prime} \wedge (P_2)^{\sigma^\prime,\tau^\prime,r^\prime} \wedge \Psi^{\sigma^\prime,\tau^\prime, r^\prime}(\mathcal{F})$ is true and by \cref{theo:predicates_frechet} the second claim follows. The cardinality of $P^\prime$ is determined by the cardinalities of $\mathcal{I}_3, \dots \mathcal{I}_6$, which in turn can be bounded as stated in the theorem statement.
\end{proof}

\begin{proof}[Proof of \cref{lem:inner_product}]
    In the following, we assume that $\lVert p \rVert = \lVert q \rVert = 1$. To prove the upper bound, observe that
    \begin{align*}
        2 \langle f(p), f(q) \rangle & = \lVert f(p) - f(0) \rVert^2 + \lVert f(q) - f(0) \rVert^2 - \lVert f(p) - f(q) \rVert^2 \\
        & \leq (1+\epsilon)^2\lVert p \rVert^2 + (1+\epsilon)^2\lVert q \rVert^2 - (1-\epsilon)^2 \lVert p - q \rVert^2 \\
        & \leq 2 \langle p, q \rangle + 6 \epsilon + 2\epsilon \lVert p - q \rVert^2 \\
        & \leq 2 \langle p, q \rangle + 6 \epsilon + 2\epsilon(\lVert p \rVert + \lVert q \rVert)^2 = 2 \langle p, q \rangle + 14 \epsilon,
    \end{align*}
    where the last inequality follows from the triangle inequality. To prove the lower bound, observe that
    \begin{align*}
        2 \langle f(p), f(q) \rangle & \geq (1-\epsilon)^2 \lVert p \rVert^2 + (1-\epsilon)^2 \lVert q \rVert^2 - (1+\epsilon)^2\lVert p - q \rVert^2 \\
        & \geq 2 \langle p, q \rangle - 2\epsilon \lVert p \rVert^2 - 2\epsilon \lVert q \rVert^2 - 3 \epsilon \lVert p - q \rVert^2 \\
        & \geq 2 \langle p, q \rangle - 4 \epsilon - 3\epsilon(\lVert p \rVert + \lVert q \rVert)^2 = 2 \langle p, q \rangle - 16 \epsilon,
    \end{align*}
    where the last inequality again follows from the triangle inequality. Using the linearity of the dot product and $f$, we have that $\langle f(p), f(q) \rangle = \lVert p \rVert \cdot \lVert q \rVert \cdot \langle f(p/\lVert p \rVert), f(q/\lVert q \rVert) \rangle$, which yields the claim.
\end{proof}

\begin{proof}[Proof of \cref{theo:simpl_embedding}]
    We use the approach from \cite[Lemma 7.1]{DBLP:conf/soda/BuchinDGHKLS19}. Here, the algorithms by Imai and Iri \cite{imai_polygonal_1988} and Alt and Godau \cite{alt_godau} are combined to obtain a vertex-restricted simplification. In detail, for the given curve $\tau$ with vertices $v^\tau_1, \dots, v^\tau_{\lvert \tau \rvert}$ a directed graph $G(\tau)$ is constructed. The vertices of the graph are the vertices of $\tau$ and it has an edge $(v^\tau_i, v^\tau_j)$ for $i,j \in [\lvert \tau \rvert]$ and $i < j$, assigned with weight $\df(\tau[i,j], \overline{v^\tau_i v^\tau_j})$. The simplification is determined by the path from $v^\tau_1$ to $v^\tau_{\lvert \tau \rvert}$ of $\ell$ vertices and with cost, i.e., maximum edge weight, minimized. Observe that the cost of the path is $\df(\tau, \simpl(\tau))$. The approximation factor and the running time follow from \cite[Lemma 7.1]{DBLP:conf/soda/BuchinDGHKLS19} (and by incorporating the dimension, which is assumed to be constant in these works).
    \\
    
    Let $\tau \in T$ and consider $\simpl(F(\tau))$ returned by the above algorithm. There exists a vertex-restricted $\ell$-simplification $\tau^\prime$ of $\tau$, such that $\simpl(F(\tau)) = F(\tau^\prime)$. 
    
    Now, since $(1-\epsilon) \df(\tau[i,j], \overline{v^\tau_i, v^\tau_j}) \leq \df(F(\tau[i,j]), F(\overline{v^\tau_i v^\tau_j}) \leq (1+\epsilon) \df(\tau[i,j], \overline{v^\tau_i, v^\tau_j})$ by \cref{coro:JLembedding}, it may be that the minimum cost path in $G(F(\tau))$ from $f(v^\tau_1)$ to $f(v^\tau_{\lvert \tau \rvert})$ of $\ell$ vertices is by a factor of $(1-\epsilon)$ cheaper than the corresponding path in $G(\tau)$, and that the path in $G(F(\tau))$ corresponding to the minimum cost path in $G(\tau)$ from $v^\tau_1$ to $v^\tau_{\lvert \tau \rvert}$ of $\ell$ vertices is by a factor of $(1+\epsilon)$ more expensive. Thus, $\df(\tau, \tau^\prime) \leq (1+\epsilon)/(1-\epsilon) \df(\tau, \simpl(\tau)) \leq (1+4\epsilon) \df(\tau, \simpl(\tau))$, for the given range of $\epsilon$. We conclude that $\tau^\prime$ is a $(4+16\epsilon)$-approximate minimum-error $\ell$-simplification of $\tau$.
    \\
    
    The remainder follows by the embedding guarantee of \cref{coro:JLembedding}.
\end{proof}

\begin{proof}[Proof of \cref{theo:klcenter}]
    For $i \in [k]$ we denote by $C_i$ the set of centers computed after the $i$\textsuperscript{th} iteration of the algorithm and let $r_i = \max_{\tau \in T} \min_{c \in C_i} \df(F(\tau), c)$. Clearly, $r_1 \geq \ldots \geq r_k$ and $r_k$ is the cost of the solution $C_k$. Furthermore, we let $C^\ast = \{c^\ast_1, \dots, c^\ast_k\}$ denote an optimal set of centers with cost $r^\ast$ and for $i \in [k]$ we let $T^\ast_i = \{\tau \in T \mid \forall j \in [k]: \df(\tau,c^\ast_i) \leq \df(\tau, c^\ast_j)\}$ denote the $i$\textsuperscript{th} optimal cluster, where we assume that ties are broken arbitrarily, such that $T^\ast_1, \dots, T^\ast_k$ form a partition of $T$.
    \\
    
    We prove the upper bound. Let $\sigma \in T$ be a curve that maximizes $\min_{c \in C_k} \df(F(\tau), c)$ among all $\tau \in T$, i.e., $\min_{c \in C_k} \df(F(\sigma), c) = r_k$ and let $C_{k+1} = C_{k} \cup \{\simpl(F(\sigma))\}$. By the pigeonhole principle there are two curves $c = \simpl(F(\tau)), c^\prime = \simpl(F(\tau^\prime)) \in C_{k+1}$, such that $\tau$ and $\tau^\prime$ that lie in the same optimal cluster $T^\ast_j$. W.l.o.g. assume that $c$ is added in an earlier iteration of the algorithm. We have 
    \begin{align*}
        r_k \leq & \df(c, F(\tau^\prime)) \leq \df(\simpl(F(\tau)), F(\tau)) + \df(F(\tau), F(\tau^\prime)) \\
        \leq & (1+\epsilon)(\df(\tau,c^\ast_j) + \df(c^\ast_j, \tau^\prime)) + \df(\simpl(F(\tau)), F(\tau)) \\
        \leq & 2(1+\epsilon) r^\ast + \df(\simpl(F(\tau)), F(\tau)) \leq (6+38\epsilon) r^\ast, 
    \end{align*}
    where the last inequality follows from \cref{theo:simpl_embedding} and by the fact that $\simpl(\tau)$ is a $(4+16\epsilon)$-approximate minimum-error $\ell$-simplification of $\tau$, hence
    \begin{align*}
        \df(\simpl(F(\tau)), F(\tau)) & \leq (1+\epsilon) \cdot \df(\simpl(\tau), \tau)  \leq (4+16\epsilon) (1+\epsilon) \df(c^\ast_j, \tau) \\
        & \leq (4+16\epsilon) (1+\epsilon) r^\ast \leq (4+36\epsilon)r^\ast.
    \end{align*}
    
    We prove the lower bound. For $i \in [k]$, let $c_i = \simpl(F(\tau_i))$ denote the single element of $C_i \setminus C_{i-1}$, where we let $C_0 = \emptyset$, and let $\sigma \in T$ be a curve that maximizes $\min_{i \in [k]} \df(F(\tau), c_i)$ among all $\tau \in T$. Thus, $\min_{c \in C_k} \df(F(\sigma), c) = r_k$. 
    \\
    
    Now, let $C = \{\simpl(\tau_1), \dots, \simpl(\tau_k)\}$ and $\sigma^\prime \in T$ be a curve that maximizes $\min_{i \in [k]} \df(\tau, \simpl(\tau_i))$ among all $\tau \in T$. Also, let $r = \df(\sigma^\prime, \simpl(\tau_i))$. Clearly, $r \geq r^\ast$ must hold, since $C^\ast$ is an optimal solution, so $C$ must have equal or larger cost.
    \\
    
    Furthermore, by \cref{theo:simpl_embedding} and by definition of $\sigma$ and $\sigma^\prime$ it holds that \[(1+\epsilon) \min_{i \in [k]} \df(\simpl(\tau_i), \sigma) \geq \min_{i \in [k]} \df(c_i, F(\sigma)) \geq \min_{i \in [k]} \df(c_i, F(\sigma^\prime)) \geq (1-\epsilon) \min_{i \in [k]} \df(\simpl(\tau_i), \sigma^\prime).\]
    We have
    \begin{align*}
        r_k = & \min_{i \in [k]} \df(c_i, F(\sigma)) = \min_{i \in [k]} \df(\simpl(F(\tau_i)), F(\sigma)) \geq (1-\epsilon) \min_{i \in [k]} \df(\simpl(\tau_i), \sigma) \\
        & \geq \frac{(1-\epsilon)^2}{(1+\epsilon)} r \geq (1-3\epsilon)r^\ast,
    \end{align*}
    where the first inequality follows from \cref{theo:simpl_embedding}.
    \\
    
    We now discuss the running time. First, exactly $k$ simplifications are computed during the execution of the algorithm. This has running time $O(k \cdot d^\prime m^3 \log m)$ by \cref{theo:simpl_embedding}. Then, in $k$ rounds, the algorithm computes $i \cdot n$ Fr\'echet distances to the $i$ centers already computed using Alt and Godau's algorithm. This takes time $O(k^2 n d^\prime m \log m)$.
\end{proof}

\section{Unrestricted medians}
\label{appendixunrestrictedmedians}
We begin by lower bounding the cost of the optimal solution on the set of projected curves.   

\begin{lemma}
\label{lemma:continuousmedianlowerbound}
Let $T$ be a set of $n$ polygonal curves in $\mathbb{R}^d$ of complexity at most $m$ each. Let $\ell \geq m$ and let $F$ be the embedding of \cref{coro:JLembedding} for $T$, with parameter $\epsilon \in (0,1)$. Then,  
\[
r^{\ast}\leq \frac{2}{1-\epsilon}\cdot r_f^{\ast},
\]
where $r^{\ast}$ is the cost of an  optimal solution to the $(k,\ell)$-median problem on $T$, and $r_f^{\ast}$ is the cost of an optimal solution to the  $(k,\ell)$-median problem on $F(T)$.  
\end{lemma}
\begin{proof}

Let $\{c_1^{f},\ldots,c_k^{f}\}$ be an optimal $(k,\ell)$-median solution on $F(T)$, let $T_1^f,\ldots, T_k^f$ be the corresponding subsets (clusters) of $F(T)$ associated with them and let $T_i = F^{-1}(T_i^f)$. 
By an averaging argument, for each $i\in[k]$ there exists a $\sigma_i\in T_i $ such that $\df(F(\sigma_i), c_i^f)\leq \frac{1}{|T_i|}\cdot \sum_{\tau\in T_i}\df(F(\tau),c_i^f)$. Then, 
\begin{align*}
r^{\ast} &\leq \sum_{i=1}^k \sum_{\tau\in T_i} \d_F(\tau,\sigma_i) \leq \frac{1}{1-\epsilon}\cdot  \sum_{i=1}^k \sum_{\tau\in T_i} \df(F(\tau),F(\sigma_i)) \\&\leq 
\frac{1}{1-\epsilon}\cdot
\sum_{i=1}^k \sum_{\tau\in T_i} \left(\df(F(\tau),{c}_i^f)+\df(c_i^f,F(\sigma_i))\right) \\&\leq 
\frac{1}{1-\epsilon}\cdot \left( r_f^{\ast}+
\sum_{i=1}^k \sum_{\tau\in T_i} \frac{1}{|T_i|}\cdot \sum_{\tau\in T_i}\df(F(\tau),c_i^f)\right) \leq \frac{2}{1-\epsilon} \cdot r_f^{\ast}.
\end{align*}
\end{proof}

Now, if we add the assumption that distances between the optimal medians in the original space and the input curves are approximately preserved, we obtain our main result.
\begin{theorem}
\label{theo:unremedians}
Let $T$ be a set of $n$ polygonal curves in $\RR^d$ of complexity at most $m$ each. Let $\ell\geq m$ and let $C^{\ast}$ be an optimal solution to the $(k,\ell)$-median problem with cost $r^{\ast}$. Let $F$ be the embedding of \cref{coro:JLembedding} for $T\cup C^{\ast}$, with parameter $\epsilon \in (0,1)$. Then,  
\[
\frac{1-\epsilon}{2}\cdot r^{\ast}\leq  r_f^{\ast} \leq (1+\epsilon)\cdot r^{\ast},
\]
where $ r_f^{\ast}$ is the optimal cost of the $(k,\ell)$-median clustering problem with input $F(T)$. 
\end{theorem}
\begin{proof}
\cref{lemma:continuousmedianlowerbound} implies $\frac{1-\epsilon}{2}\cdot r^{\ast}\leq  r_f^{\ast} $. Moreover, since $F$ is as in \cref{coro:JLembedding},  it  satisfies $\forall \tau, \sigma \in T \cup C^{\ast}: (1-\epsilon) \df(\tau,\sigma) \leq \df(F(\tau),F(\sigma)) \leq (1+\epsilon)\df(\tau,\sigma)$, we conclude that $F(C^{\ast})$ is a solution with cost at most $(1+\epsilon) r^{\ast}$. Therefore $r_f^{\ast}\leq (1+\epsilon)\cdot r^{\ast}$. 
\end{proof}
Finally,  we can apply a JL transform to obtain the following result on dimension reduction of curves.
\unrestrictedcorollary*
\begin{proof}
The theorem follows by combining \cref{JLlemma}, \cref{coro:JLembedding} and  \cref{theo:unremedians}. In particular, we apply \cref{JLlemma} for a set of $O((n+k)^2 \cdot \ell^3)$ points determined by $T \cup C^{\ast}$, where $C^\ast$ is an optimal $(k,\ell)$-median solution for $T$, to obtain the $(1\pm \epsilon)$-embedding required by \cref{coro:JLembedding}. Then, \cref{coro:JLembedding} combined with \cref{theo:unremedians}, imply the statement. 
\end{proof}

\section{Restricted median}
\label{appendixrestrictedmedians}
We start with a lower bound on the optimal cost of the embedded curves. 
\begin{lemma}
\label{lemma:continuousmedianlowerboundrestricted}
Let $T$ be a set of $n$ polygonal curves in $\RR^d$ of complexity at most $m$ each and let $k \in \mathbb{N},\ell \in \NN_{\geq 2}$. Let $F$ be the embedding of \cref{coro:JLembedding} for $T \cup S(T)$, with parameter $\epsilon \in (0,1)$,
where $S(T)$ is the set of all vertex-restricted $\ell$-simplifications of all polygonal curves in $T$. 
Then,  
\[
r^{\ast}\leq \frac{6(1+\epsilon)}{1-\epsilon}\cdot r_f^{\ast},
\]
where $r^{\ast}$ is the cost of an  optimal solution to the $(k,\ell)$-median problem on $T$, and $r_f^{\ast}$ is the cost of an optimal solution to the  $(k,\ell)$-median problem on $F(T)$.  
\end{lemma}
\begin{proof}
Let $\{c_1^{f},\ldots,c_k^{f}\}$ be an optimal $(k,\ell)$-median solution for $F(T)$, let $T_1^f,\ldots, T_k^f$ be the corresponding subsets (clusters) of $F(T)$ associated with them and let $T_i = F^{-1}(T_i^f)$. 
By an averaging argument, for each $i\in[k]$ there exists a curve  $\sigma_i\in T_i $ such that $\df(F(\sigma_i), c_i^f)\leq \frac{1}{|T_i|}\cdot \sum_{\tau\in T_i}\df(F(\tau),c_i^f)$. Let 
$\tilde{\sigma_i}$ be an optimal vertex-restricted minimum-error $\ell$-simplification of $\sigma_i$ and let 
$\tilde{\sigma}_i^f$ be an optimal vertex-restricted minimum-error $\ell$-simplification of $F(\sigma_i)$. Then, 
\begin{align}
r^{\ast} &\leq \sum_{i=1}^k \sum_{\tau\in T_i} \df(\tau,\tilde{\sigma}_i) \nonumber\\& \leq
 \sum_{i=1}^k \sum_{\tau\in T_i} \left(\df(\tau,{\sigma}_i)+\df({\sigma}_i,\tilde{\sigma}_i) \right) \label{eq:trineq1}
 \\& \leq
\frac{1}{1-\epsilon}\cdot  \sum_{i=1}^k \sum_{\tau\in T_i} \left(\df(F(\tau),F(\sigma_i)) +\df(F(\sigma_i),F(\tilde{\sigma}_i)) \right) \nonumber
\\&\leq 
\frac{1}{1-\epsilon}\cdot
\sum_{i=1}^k \sum_{\tau\in T_i} \left(\df(F(\tau),{c}_i^f)+\df(c_i^f,F(\sigma_i))+\df(F(\sigma_i),F(\tilde{\sigma}_i))\right) \label{eq:trineq2}
\\&\leq 
\frac{1}{1-\epsilon}\cdot
\sum_{i=1}^k \sum_{\tau\in T_i} \left(\df(F(\tau),{c}_i^f)+\df(c_i^f,F(\sigma_i))+\frac{1+\epsilon}{1-\epsilon}\cdot \df(F(\sigma_i),\tilde{\sigma}_i^f)\right) \label{eq:simplpres}\\
&\leq 
\frac{1+\epsilon}{(1-\epsilon)^2}\cdot
\sum_{i=1}^k \sum_{\tau\in T_i} \left(\df(F(\tau),{c}_i^f)+5\cdot \df(c_i^f,F(\sigma_i))\right) \label{eq:weaksimpl}\\
& \leq 
\frac{1+\epsilon}{(1-\epsilon)^2}\cdot \left( r_f^{\ast}+
5\cdot \sum_{i=1}^k \sum_{\tau\in T_i} \frac{1}{|T_i|}\cdot \sum_{\tau\in T_i}\df(F(\tau),c_i^f)\right)\nonumber\\
&\leq \frac{6\cdot(1+\epsilon)}{(1-\epsilon)^2} \cdot r_f^{\ast}\nonumber,
\end{align}
where in (\ref{eq:trineq1}), and (\ref{eq:trineq2}), we apply the triangle inequality,
(\ref{eq:simplpres}) follows by the fact that $F$ approximately preserves distances to all vertex-restricted  $\ell$-simplifications, and  (\ref{eq:weaksimpl}) follows from \cite[Lemma 7.1]{DBLP:conf/soda/BuchinDGHKLS19}  which states that there exists a vertex-restricted $\ell$-simplification which is within distance at most $4$ times that of any non-restricted $\ell$-simplification.
\end{proof}

By  assuming that distances between the input curves, the optimal medians and all vertex-restricted $\ell$-simplifications are approximately preserved, we obtain the following theorem.
\begin{theorem}

Let $T$ be a set of $n$ polygonal curves in $\RR^d$ of complexity at most $m$ each and let $k \in \NN, \ell \in \NN_{\geq 2}$. Let $F$ be the embedding of \cref{coro:JLembedding} for $T \cup S(T)\cup C^{\ast}$ with parameter $\epsilon \in (0,1)$, where $S(T)$ is the set of all vertex-restricted $\ell$-simplifications of all polygonal curves in $T$ and $C^{\ast}$ is an optimal $(k,\ell)$-median solution for $T$. 
Then,  
\[
\frac{1-\epsilon}{6\cdot(1+\epsilon)}\cdot r^{\ast}\leq  r_f^{\ast} \leq (1+\epsilon)\cdot r^{\ast},
\]
where $r^{\ast}$ is the cost of an  optimal solution to the $(k,\ell)$-median problem on $T$, and $r_f^{\ast}$ is the cost of an optimal solution to the  $(k,\ell)$-median problem on $F(T)$.
\end{theorem}
\begin{proof}
\cref{lemma:continuousmedianlowerboundrestricted} implies $\frac{1-\epsilon}{6\cdot(1+\epsilon)}\cdot r^{\ast}\leq  r_f^{\ast} $. Moreover, since $F$ is as in \cref{coro:JLembedding},  it  satisfies $\forall \tau, \sigma \in T \cup C^{\ast}: (1-\epsilon) \df(\tau,\sigma) \leq \df(F(\tau),F(\sigma)) \leq (1+\epsilon)\df(\tau,\sigma)$, we conclude that $F(C^{\ast})$ is a solution with cost at most $(1+\epsilon) r^{\ast}$. Therefore $r_f^{\ast}\leq (1+\epsilon)\cdot r^{\ast}$. 
\end{proof}
Finally, we can effectively reduce the dimension by applying a JL transform. 
\restrictedcorollary*
\begin{proof}
The result follows by combining \cref{JLlemma}, \cref{coro:JLembedding} and  \cref{theo:unremedians}. In particular, we apply \cref{JLlemma} for a set of $O(((n+nm^{\ell}+k)^2 \cdot m^3)$ points determined by $T \cup S(T)\cup C^{\ast}$,  
where $S(T)$ is the set of all vertex-restricted $\ell$-simplifications of all polygonal curves in $T$,  and $C^{\ast}$ is an optimal $(k,\ell)$-median solution for $T$, to obtain the $(1\pm \epsilon)$-embedding required by \cref{coro:JLembedding}. Then, \cref{coro:JLembedding} combined with \cref{theo:unremedians}, imply the statement. 
\end{proof}

\end{document}